\documentclass[a4paper,11pt]{article}
\usepackage{graphicx}
\usepackage{amsfonts,amsmath,amsthm,amssymb}
\usepackage{multirow}
\usepackage{booktabs}
\usepackage{tikz}
\usepackage{changes} 

\usepackage{fullpage}

\usepackage{footnote}
\makesavenoteenv{tabular}

\newcommand{\U}{\mathcal{U}}

\newcommand{\A}{\mathcal{A}}
\newcommand{\hh}{\mathcal{H}}
\newcommand{\F}{\mathcal{F}}
\newcommand{\D}{\mathcal{D}}
\newcommand{\PA}{\mathcal{P}}
\newcommand{\E}{\mathbb{E}}
\newcommand{\N}{\mathbb{N}}
\newcommand{\R}{\mathbb{R}}

\newtheorem{myth}{Theorem}
\newtheorem{prop}{Proposition}
\newtheorem{lemma}{Lemma}[section]

\newtheorem{example}{Example}
\newtheorem*{ex1cont*}{Example 1 - cont}

\title{Choquet expected utility across time \\ \vspace{4mm} without lotteries}
\author{Lorenzo Bastianello\footnote{Corresponding author: \texttt{lorenzo.bastianello@u-paris2.fr}.}, Universit\'e Paris 2 Panth\'eon-Assas\\
Jos\'e Heleno Faro, INSPER S\~ao Paulo}
\date{\today\\
\textbf{Preliminary version - do not circulate}
}

\begin{document}
\setlength{\baselineskip}{1\baselineskip}
\thispagestyle{empty}

\vspace*{0.5cm}

\begin{center}
{\LARGE\bf Time discounting under uncertainty\footnote{We would like to thank Antoine Bommier, Alain Chateauneuf,  Mehmet Ismail, Asen Kochov, Marcus Pivato, Jean-Marc Tallon, St\'ephane Zuber and audiences at numerous seminars and conferences. Jos\'e Heleno Faro is also grateful for financial support from the CNPq-Brazil (Grant no. 307508/2016-1).}}
\end{center}

\vspace*{.2cm}
\begin{center}
\begin{tabular}{ccc}
{\large\sc Lorenzo Bastianello\footnote{Corresponding author.}} & $\qquad$ & {\large\sc Jos\'e Heleno Faro}  \\
{\footnotesize [lorenzo.bastianello@u-paris2.fr]} & & {\footnotesize [josehf@insper.edu.br]} \\
{\small LEMMA} & & {} \\
{\small Universit{\'e} Paris 2 Panth{\'e}on-Assas} & &  {\small Insper S\~ao Paulo} \\
\end{tabular}
\end{center}

\begin{center}
\vspace*{0.4cm}
This version: March 2020
\vspace*{0.4cm}
\end{center}

\begin{abstract}
We study  intertemporal decision making under uncertainty. We fully characterize discounted expected utility in a framework \textit{\`a la} Savage. Despite the popularity of this model, no characterization is available in this setting. The concept of stationarity, introduced by Koopmans for deterministic discounted utility, plays a central role for both attitudes towards time and towards uncertainty.  We show that a strong stationarity axiom characterizes discounted expected utility. When hedging considerations are taken into account, a weaker stationarity axiom generalizes discounted expected utility to Choquet discounted expected utility, allowing for non-neutral attitudes towards uncertainty.




\end{abstract}

\section{Introduction}\label{sec:intro}

When making economic decisions, agents usually need  to take into account two important dimensions: \textit{time} and \textit{uncertainty}. Consider for instance a firm that wants to implement a project that will deliver a stochastic cash flow in the future or a government that needs to decide how to allocate its budget taking into account GDP growth in the following years. In both cases, decision makers (DMs henceforth) are required to make choices that involve uncertain outcomes occurring at future dates.

One of the most popular models used both in microeconomics and macroeconomics to evaluate uncertain streams of income or consumption, is the \textit{(exponential) discounted expected utility model}. The choice problem faced by the firm or the government can be formalized in the following way.  Suppose that, at time $t=0$, Nature chooses a state of the world $\omega\in \Omega$ which is hidden to the agent. 
The state of the world determines the amounts of income that the DM will receive across  future periods. Given $\omega\in\Omega$, she gets $h_t(\omega)$ at time $t\geq 0$, where $h_t$ is a measurable function over $\Omega$. Therefore, at time $t=0$, she is facing a stochastic stream of income. The discounted expected utility model  says that the stochastic stream $h:=(h_0,h_1,\dots)$ is evaluated through the function
\begin{equation}\label{eq:DEU}
V(h)=\sum_{t=0}^{\infty}\beta^t \E_P[u(h_t)],
\end{equation}
where  $u(\cdot)$ is an instantaneous utility function, the DM's attitude towards time is described through the discount factor $\beta^t$ with $\beta\in (0,1)$ and the expectation is taken with respect to a subjective probability $P$ over $\Omega$. Mathematically, the expression in (\ref{eq:DEU}) prescribes to compute the expected utility of the random variables $h_t$, actualize its value through the discount factor $\beta^t$ and  sum up the actualized expected utilities.\footnote{There are several interesting ways to discount the future that are outside the scope of this paper. While exponential discounting may be appropriate in the case of a firm, who can set $\beta=\frac{1}{1+r}$ where $r$ is the interest rate, it may not be so compelling for other types of agents. For instance, interesting alternatives are the hyperbolic, quasi-hyperbolic and quasi-geometric discounting, see Loewenstein and Prelec \cite{LoewPrel}, Laibson \cite{laib}, Montiel Olea and Strzalecki \cite{OleaStra}, Phelps and Pollak \cite{PhelpsPollak} and Hayashi \cite{hay}. Another interesting approach appears in Chambers and Echenique \cite{Chamb} where the authors study models with multiple discount factors.  The focus of this paper is on the axiomatization of exponential discounted expected utility  and its possible generalizations under uncertainty. }

The intuitiveness and mathematical tractability of  formula (\ref{eq:DEU}) has made  discounted expected utility the most popular way to model intertemporal decisions involving  uncertainty. 
However, and quite surprisingly, to our knowledge there is no axiomatization available in the literature of the discounted expected utility model in a framework \textit{\`a la} Savage \cite{Savage}.\footnote{For an axiomatization of the discounted expected utility model with an infinite horizon \textit{under risk} one can see Theorem 1 and Theorem 2 in Epstein \cite{Epstein83}. See also Theorem 1, Proposition 2 and Proposition 3 in Peitler \cite{Peitler}.}  Our first contribution consists in providing a set of axioms that delivers the formula in (\ref{eq:DEU}). Axiomatizations enable us to understand which behavioral conditions stand behind a utility function. See Gilboa, Postlewaite, Samuelson and Schmeidler \cite{GPSS} for a recent explanation about the importance of representation results.

The key behavioral feature that delivers discounted expected utility is a \textit{stationarity} condition. Stationarity is an axiom proposed by Koopmans \cite{Koop}, \cite{Koop72} in the deterministic framework in order to characterize discounted utility. The formula of discounted utility appeared before its axiomatization and it is due to Samuelson \cite{Sam37}.  According to this model a DM evaluates a deterministic stream of income  $d:=(d_0,d_1,\dots)$ through the following utility function
\begin{equation}\label{eq:DU}
V(d)= \sum_{t=0}^{\infty} \beta^t u(d_t).
\end{equation}
Koopmans' stationarity axiom asserts that if an agent is indifferent between two streams of income, then she will remain indifferent if we shift all incomes one period ahead  and insert a common first period income. The reasoning can be done the other way round. If an agent is indifferent between two streams that are equal in the first component, then she will remain indifferent if we drop this first coordinate and advance all other incomes by one period for both streams.

To illustrate, consider the following example. A PhD candidate is  looking for an academic job. She is indifferent between two offers from university  $x$ and university $y$. The universities pay her the salaries (correlated to the monthly teaching load) listed in the table below.
\begin{center}
\begin{tabular}{lccccc} 
 & sep & oct & nov & dec & $\dots$ \\  

$x_0,x_1,\dots$  & 10 & 20 & 1 & $\dots$ & $\dots$ \\
  
$y_0,y_1,\dots$  &  5 & 23 & 2 & $\dots$ & $\dots$ \\	

%
\end{tabular}
\end{center}
Suppose that a new law postpones the beginning of the academic year to October, and that the government obliges all universities to give a compensation of $z=7$ to the newly hired professors in September. The streams of income   now look like
\begin{center}
\begin{tabular}{lccccc} 
 & sep & oct & nov & dec & $\dots$ \\  
%
%

$z,x_0,x_1,\dots$ &  7 & 10 & 20 & 1 & $\dots$ \\ 
 
$z,y_0,y_1,\dots$   &  7 & 5 & 23 & 2 & $\dots$ \\ 
\end{tabular}
\end{center}
Stationarity asserts that if the PhD candidate is indifferent between the two universities before the law, she should also stay indifferent afterwards. 

When uncertainty is involved, we introduce a stationarity axiom that characterizes discounted expected utility. A DM who is indifferent between two uncertain streams should remain indifferent if, given a period $t$, the random incomes  of both streams are shifted one period ahead starting from $t$, and a stochastic amount of income $h_t$ is inserted in period $t$.\footnote{The main difference with the deterministic framework is that we require the property of stationarity for every period $t$. This is needed for some technical issues about measurability that will be explained once the formal setting is introduced.} The similarity between the behavioral content of our axiom and   Koopmans' original axiom should be obvious. Theorem \ref{th:EU}, one of our main results, characterizes discounted expected utility with the help of this property. 

As it turns out, the implications of the stationarity axiom in  presence of uncertainty are way stronger than in the deterministic framework. Consider a firm that needs to decide today whether to invest in solar ($s$) technology or to build a carbon ($c$) plan. In year 2021 there will be elections (this is the source of uncertainty). If democrats win (event $D$), they will subsidize firms that invested in green technologies. If republicans win (event $R$), they will subsidize traditional ways of producing energy. The firm has therefore to choose between the two streams
\begin{center}
\begin{tabular}{lccccc} 
 & 2020 & 2021 & 2022 & 2023 & $\dots$ \\  
$s_0,s_1,\dots$  & 0 & $ s_1\begin{cases}
10 & \text{if } D \\
0 & \text{if } R
\end{cases}
$ & 0 & $\dots$ & $\dots$ \\[0.5cm]
$c_0,c_1,\dots$  &  0 & $ c_1\begin{cases}
0  & \text{if } D \\
10 & \text{if } R
\end{cases}
$ & 0 & $\dots$ & $\dots$ \\	

%
\end{tabular}
\end{center}
If the manager of the firm thinks that democrats and republicans are equally likely to win the elections, then it seems reasonable to indifference between $s$ and $c$.
Suppose now that the republicans are the incumbent. For budget reasons they are obliged to postpone any subsidy in favor of the technological sector to 2022, but they promise to reduces taxes on firms in 2021 in case they win the elections. The streams now look like
\begin{center}
\begin{tabular}{lccccc} 
 & 2020 & 2021 & 2022 & 2023 &  $\dots$ \\  
$s_0,h,s_1,\dots$   & 0 & $ h\begin{cases}
0 & \text{if } D \\
7 & \text{if } R
\end{cases}
$ & $ s_1\begin{cases}
10 & \text{if } D \\
0 & \text{if } R
\end{cases}
$ & 0 & $\dots$ \\[0.5cm]
$c_0,h,c_1,\dots$   &  0 & $ h\begin{cases}
0 & \text{if } D \\
7 & \text{if } R
\end{cases}
$ & $ c_1\begin{cases}
0 & \text{if } D \\
10 & \text{if } R
\end{cases}
$ & 0 & $\dots$  \\	

%
\end{tabular}
\end{center}
If stationarity holds, than the firm should remain indifferent between $s$ and $c$.
The rationale being that in year 2021 the firm gets $h$ whatever the investment decision, and from 2022 on, the streams are the same as in the previous scenario.  

 However notice that choosing to invest in carbon looks more uncertain. In fact, if democrats win, the firm will not get anything for two consecutive years. On the other hand, investing on solar would allow the firm to hedge against any possible winner. If democrats win, the firm gets 10 in 2022 and if republicans win it gets 7 in 2021: whatever the winner, the firm has a positive result in at least one period. Choosing solar offers a possible hedge against the uncertainty of the electoral result and therefore one may expect that a prudent (or pessimistic) manager strictly prefers $s$ to $c$. 
 
  The central idea is that $h$ is positively correlated with $c_1$ and negatively correlated with $s_1$. An uncertainty averse DM would prefer negatively correlated variable following one another in order to protect herself against uncertainty. In this case we say that the DM can \textit{temporally hedge} against uncertainty. Imposing stationarity when hedging considerations can be done may be quite demanding, and this is why, under uncertainty, we name this property  \textit{Strong Stationarity}.

The example above shows a particular case of a more general concept. Loosely speaking, two random variables $\phi,\psi$ are comonotonic if they ``vary in the same direction'': when $\phi$ takes relative high (low) values, $\psi$ also takes relative high  (low) values. For instance, in the example, $h$ is comonotonic with $c_1$ but not with $s_1$. Comonotonic random variables do not make it possible to hedge against uncertainty since they are positively correlated. 
A generalization of Strong Stationarity, that we call \textit{Comonotonic Stationarity}, restricts the shifts prescribed by Strong Stationarity only to cases in which no hedging considerations can be done. Technically, we require  the stochastic income plugged at period $t$ to be comonotonic with all variables that follow in both sequences. Theorem \ref{th:CEU}, which is our second important  contribution, shows that substituting Strong Stationarity with Comonotonic Stationarity makes it possible to generalize discounted expected utility to the following representation
\begin{equation}\label{eq:CEUrepresentation}
V(h)=\int \sum_{t=0}^{\infty} \beta^t u(h_t) dv.
\end{equation}
The integral is a \textit{Choquet integral}, taken with respect to a \textit{capacity} $v$. While a precise mathematical definition will be given later, we just recall here that a capacity is a non-necessarily additive set function and the Choquet integral is a mathematical tool needed in order to take expectations with respect to it.

The Choquet expected utility model was introduced in economics in an atemporal setting in the seminal paper of Schmeidler \cite{Schmeidler2}. This model generalizes the expected utility model and solves the famous Ellsberg's paradox, see Ellsberg \cite{Ellsberg}. Loosely speaking, the paradox shows that  agents cannot quantify  uncertainty in terms of a (single) probability measure. In this sense, our model of Choquet discounted expected utility in (\ref{eq:CEUrepresentation}) parallels Schmeidler's work in the atemporal setting. It is interesting to notice that the Choquet model has been fruitfully applied to an intertemporal setting at least since Gilboa \cite{Gilboa89}. Gilboa's paper however differs from ours since uncertainty is absent and the Choquet integral is used to model aversion (or love) for variability of payments across time.

There is a rich literature concerning decision making under uncertainty, and several models have been proposed to address the Ellsberg paradox. A popular  way to take into account agents' behavior towards uncertainty  is  through the  MaxMin expected utility model of Gilboa and Schmeidler \cite{GS}.\footnote{There are several other models that deal with choice under uncertainty and may exhibit non neutral attitudes towards uncertainty. For instance the already cited Choquet expected utility model of Schmeidler \cite{Schmeidler2}, the smooth ambiguity model of Klibanoff, Marinacci and Mukerji \cite{KMM}, the variational model of Maccheroni, Marinacci and Rustichini \cite{MMR}, the confidence  model of Chateauneuf and Faro \cite{Chato-Faro} etc.}  This model, developed in the atemporal setting, says that a DM  has a set of probabilities in her mind and  takes the minimum expected utility calculated with respect to the probabilities in this set. In a recent paper, Kochov \cite{Kochov} provided a generalization of discounted expected utility (\ref{eq:DEU}) by considering an intertemporal version of the MaxMin expected utility model of Gilboa and Schmeidler \cite{GS}. Instead of a single probability measure, the DM considers a set of possible probabilities and evaluates stochastic streams by the minimum (taken with respect to the probabilities in that set) of the discounted expected utilities.  We will discuss this relevant article in the main body of our paper.

The rest of the paper is organized as follows. Section \ref{sec:framework} introduces the framework and notation. Section \ref{sec:DU} recalls Koopmans' discounted utility model and presents his axioms and result. This paves the way to Section \ref{sec:DEU} in which we introduce uncertainty and we axiomatize the discounted expected utility model. Section \ref{sec:Choquet-DEU} generalizes the previous section by characterizing the Choquet discounted utility model. Section \ref{sec:conclusion} concludes. All proofs are gathered in the Appendix.


\section{Framework and Mathematical Preliminaries}\label{sec:framework}

We borrow Kochov's \cite{Kochov} setting and notation. 
Time is discrete and identified with $\N=\{0,1,2,\dots\}$. Let $X$ be a connected, separable and first-countable topological space. It will be interpreted as the \textit{space of outcomes}. For instance, if  $X$ is a convex subset of $\R^n$, then $x\in X$ may represent a bundle of goods. Let $\Omega$ be the set of \textit{states of nature}. A \textit{filtration} $(\F_t)_t$ over $\Omega$ is a sequence of algebras such that $\F_0=\{\emptyset,\Omega\}$ and $\F_t\subseteq \F_{t+1}$ for all $t\in \N$. We denote by $\F$ the union of these algebras $\F:=\cup_t\F_t$. A \textit{stochastic process} $h  :=(h_t)_{t\in \N}$ is a sequence such that $h_t$ is a $\F_t$-measurable function from $\Omega$ to $X$ for all $t$. We sometimes call measurable functions \textit{random variables}. The following technical assumption, taken from Kochov \cite{Kochov}, restricts the set of stochastic processes that we consider. 

\medskip 

\noindent \textit{Assumption.} Stochastic processes are \textit{bounded} and \textit{finite}. Boundedness means that for each act $h$ there exists a compact set $K_h\subset X$ such that $\cup_t h_t(\Omega)\subset K_h$. Finiteness means that for each act $h$ there is a finitely generated algebra $\A_{h }\subset \F$ such that $h_t$ is $\A_{h }$-measurable for all $t\in \N$.

\medskip

\noindent We denote by $\hh$ the set of bounded and finite stochastic processes
\begin{multline*}
\hh:=\{ h  =(h_t)_{t\in \N}| h_t:\Omega\rightarrow X, h_t \text{ is } \F_t\text{-measurable } \forall t   \text{ and } h  \text{ is finite and bounded} \}.
\end{multline*}
A sequence $h  \in \hh$ will be called  \textit{act}. 
The set $\D\subset\hh$ denotes the set of \textit{deterministic acts}: we have $ d \in \D$ if and only if $d_t$ is a constant random variable for all $t\in \N$ and there exists a compact set $K_d\subseteq X$ such that $d_t\in K_d$ for all $t\geq 0$. As usual we identify $\D$ with a subset of $X^{\infty}$. For example, if $X=\R$ then $\D=l^{\infty}$, the set of real-valued bounded sequences. When  $x\in X$ and $d\in \D$, $(x,d)$ denotes the act $(x, d_0,d_1,\dots)$. Obviously the procedure can be repeated as in $(x,y,d)=(x,y,d_0,d_1,\dots)$ and so on.

A \textit{(normalized) capacity} $v$ on the  measurable space $(\Omega,\F)$ is a set function $v:\F \mapsto \mathbb{R}$ such that $v(\emptyset)=0,\,v(\Omega)=1$ and for all $A,B \in \F,\, A\subset B \Rightarrow v(A)\leq v(B)$.
Given a capacity $v$ on $(\Omega,\F)$, the \textit{Choquet integral} of a real-valued, bounded, $\F$-measurable function $f:\Omega\rightarrow \R$ with respect to $v$ is defined as
$$
\int_{\Omega}{f\,dv}:=\int_{-\infty}^0{(v(\{f\geq t\})-1)\,dt} + \int_0^{+\infty}{v(\{f\geq t\})\,dt},
$$
where $\{f\geq t\}=\{\omega\in\Omega | f(\omega)\geq t\}$. A capacity $v: \F \mapsto \mathbb{R}$ is \textit{convex} if, for all $A,B\in \F$, $v(A\cup B)+v(A\cap B)\geq v(A)+v(B)$.  A \textit{(finitely additive) probability}  $P:\F \mapsto \mathbb{R}$ is a capacity such that  $A\cap B=\emptyset$ implies $P(A\cup B)=P(A)+P(B)$. The \emph{core} of a capacity $v$ is defined by $core(v)=\{P| P \text{ is a probability s.t. } P(A)\geq v(A)\,\forall A\in \F \}$. Finally, if $P$ is a probability, then we denote the integral with respect to $P$ of a real-valued, bounded, $\F$-measurable function $f$ with the usual notation for expectation $\int_{\Omega}{f\,dP}=\E_P[f]$.

\section{The deterministic setting: Discounted Utility}\label{sec:DU}

We begin characterizing (deterministic) discounted utility through a set of axioms in the spirit of Koopmans \cite{Koop}, \cite{Koop72}. Koopmans' axioms are well known in the literature. Yet, as Bleichrodt, Rohde and Wakker \cite{Ble} put it ``his analysis is obscured by technical digressions and several inaccuracies. It is, for instance, never stated what the domain of preference is, i.e. which consumption programs are considered, and there is an unanticipated implication of bounded utility". 
Bleichrodt \textit{et al.} \cite{Ble} correct Koopmans' paper and greatly extend the domain of application of his model. In this section we slightly modify Bleichrodt \textit{et al.} \cite{Ble} in order to provide an axiomatization that is closer to the one of Koopmans and that  naturally generalizes to our framework with uncertainty. The reader who is familiar with Koopmans \cite{Koop}, \cite{Koop72} and Bleichrodt \textit{et al.} \cite{Ble}  can skip this section and proceed directly to the more general Section \ref{sec:DEU}.

 Let $\succsim$ be a complete and transitive binary relation over the set of deterministic acts $\D$. We denote $\sim$ and $\succ$ its symmetric and asymmetric part.  This relation represents the preferences of a DM deciding at time $t=0$ among deterministic $X$-valued streams. Given the preference relation  $\succsim$ over $\D$ and $x,y\in X$, we write $x\succsim y$ if $\left(  x,x,...\right)  \succsim\left(  y,y,...\right)$. Thus, preferences over $X$ agree with preferences over constant deterministic acts.

Koopmans calls his axioms \textit{postulates} and we keep the same terminology here. The postulates in this section parallel the one in Koopmans \cite{Koop72}. As Koopmans himself notices, see p. 81 in \cite{Koop72}, the postulates are not logically independent, see also Footnote \ref{fn:cont}.

The first postulate  is technical, it says that the preference relation $\succsim$ is continuous with respect to the product topology. This postulate differs from the continuity requirement in Koopmans \cite{Koop72}. He assumed $X$ to be metrizable (a stronger assumption than the one made here) and the preference relation $\succsim$ to be continuous with respect to the sup-norm topology over $\D$.\footnote{\label{fn:cont}Continuity with respect to the product topology is a stronger requirement than continuity with respect to the sup-norm topology. However, as notice in Kochov \cite{Kochov}, this allows us to drop the metrizability assumption and  it makes P.2 and P.5 redundant. This is proved formally in Proposition \ref{prop:Koop_nest}. Axiom P.2 and P.5 will be explicitly dropped in Section \ref{sec:DEU}.}

\medskip
\noindent  P.1  (\textit{Continuity}) For all compact sets $K\subset X$ and for all deterministic acts  $d'\in \D$, the sets $\{d\in K^{\infty}| d\succsim d'\}$ and $\{d\in K^{\infty}| d'\succsim d\}$ are closed  in the product topology over $K^{\infty}$.
\medskip

Postulate P.2 says that period 0 is sensitive. This implies that $\succsim$ has a non empty strict part and avoid the cases in which preferences are determined only by the tail behavior of sequences. Formally,

\medskip
\noindent P.2 (\textit{Sensitivity}) There exist $x,y\in X$, $ d\in \D$ such that $(x,d)\succ (y,d)$.
\medskip

The third postulate, Time Separability, imposes that preferences between today and tomorrow are not affected by consumption starting from period three. If two deterministic streams are equal starting from the third period onward, then their ranking does not depend on the continuations.

\medskip
\noindent  P.3   (\textit{Time Separability}) For all $x,y,x',y'\in X$ and $d,d'\in \D$, $(x,y,d)\succsim (x',y',d)$ if and only if $(x,y,d')\succsim (x',y',d')$.
\medskip

The next postulate, called K-Stationarity (for Koopmans' stationarity), plays a key role in our analysis. Let us state it formally.

\medskip
\noindent P.4 (\textit{K-Stationarity}) For all $x\in X$ and $d, d' \in \D$,  $d \succsim d' $ if and only if $(x,d )\succsim (x,d' )$.
\medskip

P.4 asserts the following. Suppose that a DM prefers a deterministic stream $d$ to $d'$. Postpone all  elements of the two sequences one period ahead ($d_0$ and $d_0'$ will be consumed in period 1, $d_1$ and $d_1'$ will be consumed in period 2 and so on) and introduce a common  period-zero consumption bundle $x$. Then she will prefer $(x,d)$ to $(x,d')$. The same reasoning can be done the other way around. If two streams have a common period zero consumption bundle, then it can be dropped, all the bundles can be shifted one period back, and preferences will not change.

Finally, K-Monotonicity (for Koopmans' monotonicity)  says, roughly, that more is preferred to less. Consider two deterministic streams $d,d'\in \D$ and suppose that, for every period $t$, the DM prefers the bundle $d_t$ rather than $d_t'$. Then K-Monotonicity implies that $d$ is preferred to $d'$. If moreover preferences are strict in at least one period of time $t$, then $d\succ d'$.

\medskip
\noindent  P.5  (\textit{K-Monotonicity}) Let $d,d'\in \D$. If $d_t\succsim d_t'$ for all $t$, then $d\succsim d'$; if moreover $d_t\succ d_t'$ for some $t$ then $d\succ d'$.
\medskip


We state now the characterization of  discounted utility.

\begin{prop}\label{prop:DU}
A preference relation $\succsim$ over $\D$ satisfies P.1-5 if and only if there exists a continuous function $u:X\rightarrow \R$ and a discount factor $\beta\in(0,1)$ such that $\succsim$ is represented by
$$
V(d)=\sum_{t=0}^{\infty}{\beta^t u(d_t)}.
$$
Moreover $\beta$ is unique and $u$ is unique up to a positive affine transformation.
\end{prop}
In order to prove Proposition \ref{prop:DU}, we will show that postulates P.1-5 imply the axioms of Bleichrodt \textit{et al.} \cite{Ble}. This will allow us to use their main result in order to get  discounted utility. The formal proof can be found in the Appendix. 

\section{Uncertainty and Discounted Expected Utility}\label{sec:DEU}


We introduce now uncertainty in the model of Section \ref{sec:DU}. Let $\succsim$ be a non-trivial (i.e. with a non-empty strict part $\succ$), complete and transitive binary relation over $\hh$. Recall that an act $h\in \hh$ is a bounded and finite stochastic process: at time $t=0$ Nature chooses a state $\omega\in \Omega$ (hidden to the DM at time $t=0$ and at any finite time $t$) and the DM receives $h_t(\omega)$ in period $t$. The relation $\succsim$ represents the preferences of the DM over acts.\footnote{Notice that the DM is choosing only at time $t=0$. Repeated choices are outside the scope of this paper as we want to stick to Koopmans \cite{Koop72} original framework. There is an extensive literature about repeated choices, recursive expected utility and dynamic consistency. A limited list of fundamental works in this area includes (in chronological order) Strotz \cite{Strotz}, Koopmans \cite{Koop},  Kreps and Porteus \cite{KP}, Epstein and Zin \cite{EZ}, Skiadas \cite{Skiadas}, Epstein and Schneider \cite{ES} and Bommier, Kochov and Le Grand \cite{BommKochLeGrand17}. }

The following three axioms  are taken  form Kochov \cite{Kochov}. We refer to them as \textit{basic axioms}.

\medskip
\noindent  \textsc{Continuity} (C) For all compact sets $K\subset X$ and for all acts $h \in\hh$ the sets $\{ d   \in K^{\infty}| d   \succsim  h  \}$ and $\{ d   \in K^{\infty}| h  \succsim  d   \}$ are closed in the product topology over $K^{\infty}$. 
\medskip

When $X$ is a subset of $\R^n$, continuity with respect to the product topology is a stronger requirement than continuity with respect to the sup-norm topology. However, axiom (C) (together with the other axioms and non-triviality of the preference relation) makes it possible to drop P.2 and P.5, see Footnote \ref{fn:cont}. Notice that (C) implies P.1.

The next axiom, Time Separability, is exactly the same as P.3 in Section \ref{sec:DU}. Several authors drop this axiom in order to get endogenous discounting, see for instance Epstein \cite{Epstein83} and Bommier, Kochov and Le Grand \cite{BommKochLeGrand19}. We keep Time Separability in order to obtain a constant exponential discount factor.

\medskip
\noindent  \textsc{Time Separability} (TS)  For all $x,y,x',y'\in X$ and $d,d'\in \D$, $(x,y,d)\succsim (x',y',d)$ if and only if $(x,y,d')\succsim (x',y',d')$.
\medskip

Before introducing Monotonicity we need a piece of notation. Let $h\in\hh$ and $\omega\in\Omega$, then $h(\omega)$ denotes the (deterministic) sequence $(h_0(\omega),h_1(\omega),\dots)\in \D$.

\medskip
\noindent \textsc{Monotonicity} (M) For all $h,g\in \hh$, if $ h  (\omega)\succsim  g   (\omega)$ $\forall \omega\in \Omega$ then $ h  \succsim  g  $.
\medskip

Under uncertainty, axiom (M) has a different interpretation than P.5. Suppose that the DM prefers the deterministic act $h(\omega)$ rather that the deterministic act $g(\omega)$  for every possible choice $\omega$ of Nature. Than (M) says that she should prefer $h$ to $g$.

The axiom Strong Stationarity represents the first important novelty of the paper. It generalizes K-Stationarity and, together with the other basic axioms, it allows us to characterize discounted expected utility.

\medskip
\noindent \textsc{Strong Stationarity} (SS) For all $t\in \N$, $f,g,h\in\hh$,
\begin{multline*}
 (h_0,\dots,h_{t-1},f_t,f_{t+1},\dots)\succsim (h_0,\dots,h_{t-1},g_t,g_{t+1},\dots) \Leftrightarrow \\
  (h_0,\dots,h_{t-1},h_t,f_t,f_{t+1},\dots)\succsim (h_0,\dots,h_{t-1},h_t,g_t,g_{t+1},\dots)
\end{multline*}

Axiom (SS) has the same behavioral interpretation of K-Stationarity. Consider two acts that are equal until period $t$ and suppose that the first act is preferred to the second one. Then we can shift all variables one period ahead starting from period $t$, introduce a common random  variable $h_t$ in period $t$ and the DM's preferences will not change. As for K-Stationarity, the reasoning can be done the other way around, i.e. the variable $h_t$ can be dropped, all the other variables can be shifted one period back, and  preferences will not change. It is not difficult to see that (SS) implies K-Stationarity (just consider (SS) over $\D$).  A caveat is in order, since acts are adapted to the filtration $(\F_t)_t$, one needs to be careful not to shift back a $\F_t$-measurable variable to period $t-1$, as this variable may not be $\F_{t-1}$-measurable. Measurability is also the reason why we need to state (SS) for all periods $t$ and not only for period 0. If one states the axiom only for shifts starting at period 0, it would be possible to insert only random variables that are measurable with respect to $\F_0$, i.e. constant random variables. We will elaborate more on this in Section \ref{sec:UncAv} when we present a very weak stationairty axiom due to Kochov \cite{Kochov}. It is important to notice that (SS) plays the role that the independence axiom has in the Anscombe-Aumann model. Shifting variables one period ahead (or one period back) and inserting (or removing) a variable $h_t$ can be done \textit{independently} of the choice of the sequences and the variable $h_t$. We will come back on this in Section \ref{sec:Choquet-DEU}, after introducing the axiom of Comonotonic Stationarity.

The following proposition formally shows that the set of axioms presented here implies the axioms of Section \ref{sec:DU}.

\begin{prop}\label{prop:Koop_nest}
$(C)$, $(TS)$, $(M)$ and $(SS)$ imply P.1-5.
\end{prop}

We are ready to state the main result of this section. The basic axioms (C), (M) and (TS) together with (SS) deliver the discounted expected utility representation in (\ref{eq:DEU}).

\begin{myth}\label{th:EU}
A preference relation $\succsim$ over $\hh$ satisfies (C), (M), (TS) and (SS)  if and only if there exists a probability $P:\F\rightarrow[0,1]$, a continuous utility function $u:X\rightarrow \R$ and a discount factor $\beta\in(0,1)$ such that   $\succsim$ is represented by
$$
V(h)=\E_P\left[\sum_{t=0}^{\infty}\beta^t u(h_t)\right]
$$ 
Moreover $P$ and $\beta$ are unique and $u$ is unique up to a positive affine transformation.
\end{myth}

Theorem \ref{th:EU} follows as a corollary of the more general Theorem \ref{th:CEU}, which is stated in Section~\ref{sec:Choquet-DEU}. A last remark is in order. In the literature, the functional of Theorem \ref{th:EU} is usually written 
$$
\sum_{t=0}^{\infty}\beta^t \E_P\left[u( h_t)\right]
$$
where the sum and the expected value operator are exchanged. Clearly this  cannot be done for the functionals in (\ref{eq:CEUrepresentation}) and (\ref{eq:maxmin}).
However we show in Proposition \ref{prop:exchange} that for the discounted expected utility model of Theorem \ref{th:EU} both forms are possible.

\begin{prop}\label{prop:exchange}
For all $h\in \hh$, $\E_P\left[\sum_t \beta^t u(h_t)\right]=\sum_t\beta^t \E_P\left[u(h_t)\right]$.
\end{prop}

\section{Comonotonic Stationarity and Choquet Discounted Expected Utility}\label{sec:Choquet-DEU}

While discounted expected utility, axiomatized in Theorem \ref{th:EU}, is widely used in economic applications, we argue that its behavioral foundations are not exempt from criticisms. Example \ref{ex:como_stat_axiom} below, already mentioned in the Introduction, shows that (SS) may be too requiring (this is why we call this axiom \textit{Strong} Stationarity) as it neglects possible hedging considerations made by the DM. 

\begin{example}\label{ex:como_stat_axiom} Let $X$ be an interval of $\R$. Acts are  interpreted as stochastic flows of income. Consider an event $A\in \F_t$ (this was the event ``Democrats win the elections'' in the Introduction) and two acts $f\sim g$, with  $f=(d_0,\dots,d_{t-1},f_t,d_{t+1},\dots)$ and $g  =(d_0,\dots,d_{t-1},g_t,d_{t+1},\dots)$. These acts are constant in every period except in period $t$ in which they are defined by
$$
f_t(\omega)=\begin{cases}
10\$ & \text{ if } \omega\in A\\
0\$ & \text{ if } \omega\in A^c
\end{cases}
\text{ and }
g_t(\omega)=\begin{cases}
0\$ & \text{ if } \omega\in A\\
10\$ & \text{ if } \omega\in A^c.
\end{cases}
$$
If $\succsim$ has the discounted expected utility  representation of Theorem \ref{th:EU} (with $u(0)=0$) then  $f\sim g$ implies $\beta^t \E_P[u(f_t)]=\beta^t \E_P[u(g_t)]$ and hence $P(A)=P(A^c)=\frac{1}{2}$,  i.e. being indifferent reveals that the DM thinks that the two events $A$ and $A^c$ are equally likely.\\
Consider now an act $h \in \hh$ such that 
 $$
h_t(\omega)=\begin{cases}
0\$ & \text{ if } \omega\in A\\
7\$ & \text{ if } \omega\in A^c.
\end{cases}
$$
If (SS) holds then one obtains
$$
(d_0,\dots,d_{t-1},h_t,f_t,d_{t+1},\dots)\sim (d_0,\dots,d_{t-1},h_t,g_t,d_{t+1},\dots).
$$
Of course, indifference is obvious if preference are of the discounted expected utility type.

However, as we have argued in the introduction, deducing indifference in this latter situation from $f\sim g$ is not  straightforward. We argue that we may actually observe
$$
(d_0,\dots,d_{t-1},h_t,f_t,d_{t+1},\dots)\succ (d_0,\dots,d_{t-1},h_t,g_t,d_{t+1},\dots).
$$
This happens because in case of bad luck, formally for $\omega\in A^c$, the act on the right of the preference relation makes the DM ``poor" for 2 consecutive dates (she gets 0\$ in $t$ and 0\$ in $t+1$). Whereas by choosing the act on the left she can be sure that she will be``rich" in at least one period. Introducing $h_t$ in front of $f_t$ gives a \textit{temporal hedge} to the DM against any choice of nature. If one is willing to take into account DMs' preferences for temporal hedging, axiom (SS) should be weakened.

\end{example}

Recall that two random variables $f,g$ from $\Omega$ to $X$ are \textit{comonotonic} if there is no $\omega$ and $\omega'$ in $\Omega$ such that $f(\omega)\succ f(\omega')$ and $g(\omega')\succ g(\omega)$.\footnote{Preferences over $X$ are defined in Section \ref{sec:DU}.} Notice that a constant random variable is comonotonic with any other random variable. If $X$ is a convex interval of $\R$, two random variables $f,g$ are comonotonic if for all $\omega,\omega'\in \Omega$, $\left[f(\omega)-f(\omega')\right]\cdot\left[g(\omega)-g(\omega')\right]\geq 0$ (here we implicitly assume that more money is preferred to less money). It is not difficult to check that in Example \ref{ex:como_stat_axiom} the variables $h_t$ and $g_t$ are comonotonic, whereas $h_t$ and $f_t$ are not. Therefore, putting $h_t$ and $f_t$ one after the other can lead to some hedging considerations. 

A DM can  \textit{temporally hedge} against  nature's choices whenever two consecutive random variables are not comonotonic. On the other hand, two consecutive comonotonic random variables remove any hedging consideration that could be done by the agent. 

The second and important novelty of our paper is represented by the axiom  Comonotonic Stationarity. This axiom weakens (SS) by taking into account temporal hedging. The idea is simply to restrict the set of acts on which (SS) has a bite. The formal statement follows.

\medskip
\noindent \textsc{Comonotonic Stationarity} (CS) For all $t\in \N$, $d  \in \D$ and  $f   ,g  ,h \in\hh$ such that $h_t$ is comonotonic with $f_i$ and $g_i$ for all $i\geq t$, 
\begin{multline*}
 (d_0,\dots,d_{t-1},f_t,f_{t+1},\dots)\succsim (d_0,\dots,d_{t-1},g_t,g_{t+1},\dots) \Leftrightarrow \\
  (d_0,\dots,d_{t-1},h_t,f_t,f_{t+1},\dots)\succsim (d_0,\dots,d_{t-1},h_t,g_t,g_{t+1},\dots)
\end{multline*}

 Axiom (CS) allows the same type of shifts as (SS) only when the random variable $h_t$, inserted in period $t$, is comonotonic with \textit{all} random variables after period $t$ for both acts $f   $ and $g$.  Comonotonic random variables do not allow for hedging between different periods of time. Therefore (CS) generalizes (SS) (in the sense that (SS) implies (CS)) limiting its range of action. Shifts can be performed only when no hedging considerations can be done. For instance the preferences described in Example \ref{ex:como_stat_axiom} are admissible under (CS).
 
It is interesting to stress the conceptual similarity of axiom (CS) with the axiom Comonotonic Independence of Schmeidler \cite{Schmeidler2}. Comonotonic Independence restricts the classical Independence axiom of expected utility to comonotonic acts. In the Anscombe and Aumann \cite{AA} (atemporal) framework, acts are functions from states of the world to lotteries over $X$. The Independence axiom says that for any three acts $f,g$ and $h$ and any mixing weight $\alpha\in[0,1]$, $f\succsim g$ if and only if $\alpha f+(1-\alpha)h\succsim \alpha g+(1-\alpha)h$. Comonotonic Independence retains this property only when the act $h$ is comonotonic with both $f$ and $g$, see Schmeidler \cite{Schmeidler2}. In this latter case, no hedging can occur when $h$ is mixed with $f$ or $g$. In our framework, hedging is not achieved by probability mixing but through ``time mixing''. Hence (SS) plays the role of Independence while (CS) the one of Comonotonic Independence.

The following is the main result of this section.

\begin{myth}\label{th:CEU}
A preference relation $\succsim$ over $\hh$ satisfies (C), (M), (TS) and (CS) if and only if there exists a capacity $v:\F\rightarrow[0,1]$, a continuous  utility function $u:X\rightarrow \R$ and a discount factor $\beta\in(0,1)$ such that  $\succsim$ is represented by
$$
V(h)=\int \sum_{t=0}^{\infty} \beta^t u( h_t) dv.
$$ 
Moreover $v$ and $\beta$ are unique and $u$ is unique up to a positive affine transformation.
\end{myth}

\subsection{Uncertainty Aversion}\label{sec:UncAv}

In this section we consider the  case of uncertainty aversion.  As it will be clearer later, uncertainty aversion is akin to some form of pessimism. This justifies the name of the following axiom.

\medskip
\noindent \textsc{Pessimistic Stationarity} (PS)  The preference relation $\succsim$ satisfies (CS). Moreover for all $t\in \N$, for all $d  \in \D$ and for all $f   ,g  ,h \in\hh$ such that $h_t$ is comonotonic with $g_i$ for all $i\geq t$, 
\begin{multline*}
 (d_0,\dots,d_{t-1},f_t,f_{t+1},\dots)\succsim (d_0,\dots,d_{t-1},g_t,g_{t+1},\dots) \Rightarrow \\
  (d_0,\dots,d_{t-1},h_t,f_t,f_{t+1},\dots)\succsim (d_0,\dots,d_{t-1},h_t,g_t,g_{t+1},\dots)
\end{multline*}
It has the same interpretation of Wakker's \cite{Wakker} axiom Pessimism Independence. See also  Chateauneuf \cite{Chato}. For a pessimistic DM shifting all variables one period ahead  starting from period  $t$ and inserting a comonotonic variable in front of $g_t$ will decrease the appreciation of the act $(d_0,\dots,d_{t-1},g_t,g_{t+1},\dots)$. On the other end, the possibly non-comonotonic variable $h_t$ in front of $f_t$ will make the stream $ (d_0,\dots,d_{t-1},f_t,f_{t+1},\dots)$ more appealing since it may offer a temporal hedge. Consider again Example \ref{ex:como_stat_axiom}.

\begin{ex1cont*}
Consider the acts of Example \ref{ex:como_stat_axiom}. If DM's preferences satisfy (PS) one will actually observe the path of choices described in Example \ref{ex:como_stat_axiom}, namely
$$
f\sim g \Rightarrow (d_0,\dots,d_{t-1},h_t,f_t,d_{t+1},\dots)\succ (d_0,\dots,d_{t-1},h_t,g_t,d_{t+1},\dots).
$$
Let $\succsim$ be represented by the  Choquet discounted expected utility functional $V$ of Theorem \ref{th:CEU} (with $u(0)=0$).  First notice that $f\sim g$ implies $V(f)=\beta^t u(10) v(A)=\beta^t u(10) v(A^c)=V(g)$; and hence $v(A)=v(A^c)$.  Since $v$ it is not required to be additive, we may have $v(A)\neq\frac{1}{2}$.

Call $\hat{f}=(d_0,\dots,d_{t-1},h_t,f_t,d_{t+1},\dots)$ and $\hat{g}=(d_0,\dots,d_{t-1},h_t,g_t,d_{t+1},\dots)$ and assume for the sake of calculations $\beta u(10)>u(7)$ (the opposite case is treated similarly and yields the same conclusions). Note that $V(\hat{f})= \beta^t u(7)+[\beta^{t+1} u(10)-\beta^t u(7)]v(A) + k$ and $V(\hat{g})=[\beta^t u(7)+ \beta^{t+1} u(10)]v(A^c)+k$ where $k= \sum_{k\neq t,t+1} \beta^k u(d_k)$. Therefore $\hat{f}\succ \hat{g}$ if and only if
\begin{align*}
u(7)+[\beta u(10)- u(7)]v(A) & > [u(7)+ \beta u(10)]v(A^c) \\
 u(7)(1-v(A)) & > u(7)v(A^c)\,\,\,\,\, [\text{ since } v(A)=v(A^c)] \\
v(A)+v(A^c) &<1
\end{align*} 
and hence $v(A)<\frac{1}{2}$. If $v$ is actually convex, one gets $v(A)+v(A^c) <1$, and in this case one can observe $f\sim g$ and  $\hat{f}\succ \hat{g}$. It is not difficult to find examples in which $v$ is convex and actually $g\succ f$ and  $\hat{f}\succ \hat{g}$. As Theorem \ref{th:CEU_v_convex} shows, (PS) forces the capacity $v$ to be convex.  

\end{ex1cont*}

Notice that (PS) implies (CS) and hence Theorem \ref{th:CEU} remains valid if (CS) is replaced by (PS). Since this latter axiom is stronger, we can prove in Theorem \ref{th:CEU_v_convex} that the capacity appearing in the Choquet integral is convex.

\begin{myth}\label{th:CEU_v_convex}
A preference relation $\succsim$ over $\hh$ satisfies (C), (M), (TS) and (PS) if and only if there exists a  convex capacity $v:\F\rightarrow[0,1]$, a continuous utility function $u:X\rightarrow \R$ and a discount factor $\beta\in(0,1)$ such that $\succsim$ is represented by
$$
V(h)=\int \sum_{t=0}^{\infty} \beta^t u(h_t) dv.
$$  
Moreover $v$ and $\beta$ are unique and $u$ is unique up to a positive affine transformation.
\end{myth}

It is known that (see Schmeidler \cite{Schmeidler}) when $v$ is a convex capacity the following equality holds
\begin{equation}\label{eq:CEU_core}
\int \sum_{t=0}^{\infty} \beta^t u(h_t) dv=\min_{P\in core(v)}\E_P\left[ \sum_{t=0}^{\infty}\beta^t u(h_t)  \right]. 
\end{equation}
This equality suggests a sharp interpretation of a Choquet integral with respect a convex capacity $v$ and justifies why we call pessimist a DM with preferences as the ones in Theorem \ref{th:CEU_v_convex}. When $v$ is convex, an agent reasons as if  she computes the discounted expected utility for all probabilities in $core(v)$ and then selects the minimal one. This representation is similar to the one axiomatized by Kochov \cite{Kochov}. 

Consider the following axiom proposed by Kochov \cite{Kochov}.

\medskip
\noindent \textsc{Intertemporal Hedging} (IH) For all $t\in \N$, for all $d\in \D$ and for all $g,h\in\hh$,
\begin{multline*}
 (d_0,\dots,d_{t-1},h_t,h_{t},d_{t+2}\dots)\sim (d_0,\dots,d_{t-1},g_t,g_{t},d_{t+2}\dots)\\
 \Rightarrow  (d_0,\dots,d_{t-1},g_t,h_{t},d_{t+2}\dots)\succsim (d_0,\dots,d_{t-1},h_t,h_{t},d_{t+2}\dots)
\end{multline*}
 The interpretation of (IH) is that a DM prefers to smooth consumption through states rather than through time. This in turns implies that she is pessimistic \textit{vis-\`a-vis} Nature's choice of the state of the world. We can notice in fact that the act $(d_0,\dots,d_{t-1},g_t,h_{t},d_{t+2}\dots)$ allows for a temporal mix that may provide some hedging against uncertainty. As explained by Kochov \cite{Kochov} this axiom is the intertemporal counterpart to the Ambiguity Aversion axiom of Gilboa and Schmeidler \cite{GS}.

We show now that the representation of Theorem \ref{th:CEU_v_convex} can be obtained also using (IH) and weakening (PS) to (CS).

\begin{myth}\label{th:CEU_v_convex_IH}
A preference relation $\succsim$ over $\hh$ satisfies (C), (M), (TS), (CS) and (IH) if and only if there exists a  convex capacity $v:\F\rightarrow[0,1]$, a continuous  utility function $u:X\rightarrow \R$ and a discount factor $\beta\in(0,1)$ such that 
$\succsim$ is represented by
$$
V(h)=\int \sum_{t=0}^{\infty} \beta^t u(h_t) dv.
$$   
Moreover $v$ and $\beta$ are unique and $u$ is unique up to a positive affine transformation.
\end{myth}

The utility function in Theorem \ref{th:CEU_v_convex} and Theorem \ref{th:CEU_v_convex_IH} is a particular case of the intertemporal MaxMin model studied by Kochov \cite{Kochov}
\begin{equation}\label{eq:maxmin}
V(h)=\min_{P\in \PA}\E_P\left[ \sum_{t=0}^{\infty}\beta^t u(h_t)  \right]
\end{equation}
where $\PA$ is a convex and closed (with respect to the weak* topology) set of  probabilities. Actually, if one sets $\PA=core(v)$, the utilities in (\ref{eq:CEU_core}) and (\ref{eq:maxmin}) are exactly  the same. 

If one is willing to obtain Kochov's \cite{Kochov} Theorem 1, in which he characterizes the utility function in (\ref{eq:maxmin}), then   (CS) should be weakened to the following axiom (called Stationarity (S) in Kochov \cite{Kochov}).

\medskip
\noindent \textsc{Kochov Stationarity} (KochS) For all $x\in X$ and $h ,g \in\hh$, $h\succsim g$ if and only if $(x,h)\succsim (x,g)$. 
\medskip

This axiom is weaker than (CS) since the outcome $x\in X$ can be identified with a constant random variable, which is comonotonic with all other variables. Notice also that it is not needed to state the axiom for all periods of time $t\in \N$. A constant random variable is $\F_t$-measurable for all $t\in \N$ and therefore there are no measurability concerns. The interpretation of this axiom is  the same as for K-Stationarity. Kochov \cite{Kochov} underlines an interesting parallel between (KochS) and the Certainty Independence axiom of Gilboa and Schmeidler \cite{GS}. 

Proposition \ref{prop:station_axioms} recalls how the different stationarity axioms presented in the paper are linked one to another.

\begin{prop}\label{prop:station_axioms}
The following implications hold:
$$
(SS)\Rightarrow (PS) \Rightarrow (CS) \Rightarrow (KochS)\Rightarrow (K\text{-}Stationaity).
$$
\end{prop}

Weakening (CS) to (KochS) and adding (IH) to the other basic axioms, allow us to recover Kochov's \cite{Kochov} main result (Theorem 1, p. 245).

\begin{myth}[Kochov \cite{Kochov}]\label{th:Kochov}
A preference relation $\succsim$ over $\hh$ satisfies (C), (M), (TS), (KochS) and (IH) if and only if there exists nonempty weak*-closed convex set $\PA$ of probabilities, a continuous strictly increasing utility function $u:X\rightarrow \R$ and a discount factor $\beta\in(0,1)$ such that  $\succsim$ is represented by
$$
V(h)=\min_{P\in \PA}\E_P\left[ \sum_{t=0}^{\infty}\beta^t u(h_t) \right].
$$ 
Moreover $v$ and $\beta$ are unique and $u$ is unique up to a positive affine transformation.
\end{myth}

\section{Conclusions}\label{sec:conclusion}

In this paper we make two contributions that can be of high interest for economists  working with  problems that involve decisions through time and under uncertainty.

First, we axiomatize the discounted expected utility model in a framework without lotteries. Despite the popularity of this model, no axiomatization is available in the literature. This paper shows that a Strong Stationarity axiom is the key behavioral conditions behind discounted expected utility. Strong Stationarity  plays the same role as the Independece axiom for decisions under uncertainty in an atemporal setting.

Second, we argue that Strong Stationarity neglects agents' hedging behavior. Strong Stationarity is subject to the same critiques as the Independece axiom. We solve this problem introducing a new axiom, Comonotonic Stationarity. This latter condition allows us to generalize the discounted expected utility model to the Choquet discounted expected utility model. Our axioms have a simple interpretation, akin to the original stationarity condition of Koopmans \cite{Koop}, \cite{Koop72}. Testing these axioms in the lab is an interesting topic for future work.


\appendix
\section{Appendix}
\small

\noindent \textit{\textbf{Proof of Proposition \ref{prop:DU}}}

\begin{proof}
As mentioned in Section \ref{sec:DU},  Bleichrodt \textit{et al.} \cite{Ble} do not use the same set of axioms that we employ. Specifically they do not use axiom P.1 and axiom P.5. We list here their alternative axioms.\\
Notation: given $T\in\N$, $X_T=\{(x_0,x_1,\dots,x_T,\alpha,\alpha)| x_0,\dots,x_T,\alpha\in X\}$. Notice that, for any $T\in \N$, there is a one-to-one function from $X_T$ and the product $X^{T+1}$.

\medskip
\noindent  \textit{Ultimate Continuity (UC).} $\succsim$ is continuous (with respect to the product topology) on each set $X_T$, i.e. the sets $\{x\in X_T| x\succsim y\}$ and $\{x\in X_T| y\succsim x\}$ are closed for all $y\in X_T$.
\medskip

\medskip
\noindent  \textit{Constant equivalent (CE).} $\succsim$ satisfies constant equivalence if for all $d\in \D$ there exists a constant sequence $x_d\in \D$ such that $d\sim x_d$.
\medskip

\medskip
\noindent  \textit{Tail Robustness (TR).} $\succsim$ satisfies tail robustness if for all constant sequence $x\in \D$, if $d\succ (\prec) x$ then there exists $t\in\N$ such that $(d_0,\dots,d_T,x,x,)\succ (\prec) x$ for all $T\geq t$.
\medskip

\begin{myth}[Bleichrodt \textit{et al.} \cite{Ble}]\label{th:wakker}
Let $\succsim$ be defined over $\D'\supset\D$, a domain that contains all ultimately constant programs, then TFAE:
\begin{itemize}
\item[(i)] DU holds over $\D'$ with $u$ continuous and not constant.
\item[(ii)] $\succsim$ satisfies P.2, P.3, P.5, UC, CE, TR.
\end{itemize}
\end{myth}

In Observation 3, Bleichrodt \textit{et al.} \cite{Ble} noticed that ``Tail robustness can also be replaced by monotonicity if [$\D'$] contains only bounded programs." Boundededness means that for every $d\in \D'$ there exist $x,y\in X$ such that $x\succsim d \succsim y$. 

\noindent \textit{Remark.} The definition of $\D$ (compactness) and P.5 imply boundedness. Therefore Observation 3 of Bleichrodt \textit{et al.} \cite{Ble} applies and we only need to show that UC and CE hold.

\begin{prop}\label{prop:P1-5_to_UC_CE}
If a preference relation $\succsim$ over $\D$ satisfies P.1-5 then it satisfies UC and CE.
\end{prop}

\begin{proof}
We show that $\succsim$ satisfies UC. Fix $T\in \N$ and $y\in X_T$. We will show that the set $U_y=\{x\in X_T| x\succsim y\}$ is closed. \\
Let $(\hat{x}^n)_n$ be a sequence in $U_y$ such that $\hat{x}^n\rightarrow \hat{x}$. Notice that $\hat{x}^n=(\hat{x}^n_0,\hat{x}^n_1,\dots,\hat{x}^n_T,\hat{\alpha}^n,\hat{\alpha}^n,\dots)$. By definition  $\hat{x}^n\rightarrow \hat{x}$ if and only if $\hat{x}^n_i\rightarrow \hat{x}_i$ for all $i\in \N$, with $\hat{x}^n_i=\hat{\alpha}^n$ for all $i\geq T+1$. Since $\hat{x}^n\in X_T$,  $\hat{\alpha}^n\rightarrow \hat{\alpha}$ and this implies that $\hat{x}\in X_T$. The sets $C_i=\{\hat{x}^n_i| n\in \N\}$, $i=0,\dots,T$, and $C_{\alpha}=\{\hat{\alpha}^n| n\in \N\}$ are compact and hence $C=\left(\cup_{i=0}^TC_i\right)\cup C_{\alpha}$ is compact. Since  $\hat{x}\in \D$, there is a compact set $\hat{K}$ such that $\hat{x}_t\in \hat{K}$ for all $t\in \N$. Therefore $K=C\cup \hat{K}$ is compact and $\hat{x}^n,\hat{x}\in K^{\infty}$ for all $n\in \N$. Consider now $U=\{d\in  K^{\infty}| d\succsim y\}$. We have therefore that for all $n\in \N$, $\hat{x}^n\succsim y$ and $\hat{x}^n\in U$, moreover $\hat{x}^n\rightarrow \hat{x}$ and since $U$ is closed by P.1 we obtain $\hat{x}\succsim y$. Hence $U_y$ is closed.

We show that $\succsim$ satisfies CE. Fix $d\in \D$, and let $K_d$ be a compact set such that $d_t\in K_d$ for all $t\in \N$. By compactness, there are $x_0, x_1\in K_d$ such that $x_0\succsim d_t\succsim x_1$ for all $t\in \N$. By P.5 we have $x_0\succsim d \succsim x_1$. Consider $A=\{y\in co(K_d)| \bar{y}\succsim d\}$ and $B=\{y\in co(K_d)| d\succsim \bar{y}\}$, where $\bar{y}$ denotes the constant sequence $\bar{y}=(y,y,\dots)$ and $co(K_d)$ is the convex hull of $K_d$. By P.1, $A$, and $B$ are closed  and since $x_0\in A$ and $x_1\in B$ they are both non empty. By connectedness of $X$ we have that $co(K_d)$ is connected and therefore there exists $x_d\in co(K_d)$ such that $x_d\sim d$.  
\end{proof}
Proposition \ref{prop:P1-5_to_UC_CE} and Theorem \ref{th:wakker} of Bleichrodt \textit{et al.} \cite{Ble} imply Proposition \ref{prop:DU}.
\end{proof}

\noindent \textit{\textbf{Proof of Proposition \ref{prop:Koop_nest}}}

\begin{proof}
It easy to see that $(C)\Rightarrow P.1$, $(TS)\Leftrightarrow P.3$, $(SS)\Rightarrow P.4$. 

\begin{lemma}
$(C)$, $(M)$ and $(SS)$ imply $P.2$.
\end{lemma}
\begin{proof}
See Lemma 5 in Kochov \cite{Kochov}.
\end{proof}

\begin{lemma}
$(C)$ and $(SS)$ imply $P.5$.\footnote{A similar statement is claimed without proof in (the proof of) Lemma 7 in Kochov \cite{Kochov}.}
\end{lemma}
\begin{proof}
Recall that if $x,y\in X$, we define $x\succsim y \Leftrightarrow (x,x,\dots)\succsim (y,y,\dots)$. We introduce the notation $({}_{n}x,d)=(\underbrace{x,\dots,x}_{n-\text{times}},d)$.

\textit{Claim 1.} Let $x,y\in X$ and $d\in \D$. Then $x\succsim y\Leftrightarrow (x,d)\succsim (y,d)$.\\
($\Rightarrow$) By $(SS)$, $(x,x,\dots)\succsim (y,y,\dots) \Leftrightarrow (x,d_0,x,\dots)\succsim (y,d_0,y,\dots)$. By induction, we find that for all $t\in \N$, $(x,x,\dots)\succsim (y,y,\dots)\Leftrightarrow (x,d_0,\dots,d_t,x,x\dots)\succsim (y,d_0,\dots,d_t,y,y\dots)$. The sequence $(x,d_0,\dots,d_t,x,x\dots)$ converges to $(x,d)$ as $t\rightarrow\infty$ therefore $(C)$ implies $(x,d)\succsim(y,d)$. \\
($\Leftarrow$) By $(SS)$, $(x,d)\succsim (y,d)\Leftrightarrow (x,x,d)\succsim (y,x,d)$. If $(y,y,d)\succ (y,x,d)$, then by $(SS)$ we would have $(y,d)\succ (x,d)$, a contradiction. Therefore $(x,x,d)\succsim (y,x,d)\succsim (y,y,d)$. By induction we find that for all $n\in \N$, $({}_{n}x,d)\succsim ({}_{n}y,d)$ and by $(C)$, $(x,x,\dots)\succsim (y,y,\dots)$.\\
 This conclude the prof of Claim 1.

\textit{Claim 2.} Let $x,y\in X$ and $d\in \D$. Then $x\succ y\Leftrightarrow (x,d)\succ (y,d)$.\\
($\Rightarrow$) By Claim 1, we know that  $(x,d)\succsim (y,d)$. Suppose that $(x,d)\sim (y,d)$. Then by $(SS)$, $(x,x,d)\sim (y,x,d)$ and $(y,x,d)\sim (y,y,d)$, hence $(x,x,d)\sim  (y,y,d)$. By induction we find  that for all $n\in \N$, $({}_{n}x,d)\sim ({}_{n}y,d)$  and by $(C)$, $(x,x,\dots)\sim (y,y,\dots)$. This contradict  $x\succ y$. Therefore we must have $(x,d)\succ (y,d)$.\\
($\Leftarrow$) Suppose $(x,d)\succ (y,d)$ but $y\succsim x$. Then by Claim 1, $(y,d)\succsim (x,d)$, which is absurd. Hence $x\succ y$.

Let $d,d'\in \D$ such that $d_t\succsim d_t'$ for all $t\in \N$. Fix $n\in \N$, since $d_n\succsim d_n'$, $(SS)$ then implies $(d_{n-1},d_n,d_n,\dots)\succsim (d_{n-1},d_n',d_n',\dots)$ and Claim 1 yields $(d_{n-1},d_n,d_n,\dots)\succsim (d_{n-1}',d_n',d_n',\dots)$. Applying $n$ times this reasoning we obtain 
$$
(d_0,d_1,\dots,d_n,d_n,\dots)\succsim (d_0',d_1',\dots,d_n',d_n',\dots)
$$
Since $n$ was arbitrary, this is true for all $n\in\N$. Letting $n$ going to infinity, we notice that the sequence on the left hand-side of the preference converges to $d$, and the one to the right hand-side converges to $d'$ in the product topology. Then $(C)$ implies $d\succsim d'$.\\
Suppose now that in addition $d_t\succ d_t'$ for some $t$. By Claim 2 we have $(d_t,d_{t+1},\dots)\succ (d_t',d_{t+1},\dots)$. Applying $t$ times $(SS)$ we obtain 
$$
(d_0,\dots,d_{t-1}, d_t,d_{t+1},\dots)\succ (d_0,\dots,d_{t-1}, d_t',d_{t+1},\dots).
$$
By Claim 1, for all $n\geq t$ $(d_0,\dots,d_{t-1}, d_t',d_{t+1},\dots)\succsim (d_0',\dots,d_n',d_{n+1}\dots)$. Since the latter sequence converge to $d'$, by $(C)$ we have $(d_0,\dots,d_{t-1}, d_t',d_{t+1},\dots)\succsim d'$ and hence $d\succ d'$.
\end{proof}
Therefore we proved that $(C)$, $(TS)$, $(M)$ and $(SS)$ imply $P$.1-5.
\end{proof}

\noindent \textit{\textbf{Proof of Theorem \ref{th:CEU}}}

\begin{proof}
We first prove necessity of the axioms. Monotonicity (M) and Time Separability (TS) follow from the properties of monotonicity and comonotonic additivity of the Choquet integral. Continuity (C) is proved in Kochov \cite{Kochov} (p. 252).  Given a compact set $K\subseteq X$, Kochov shows that the fact that $V(d)=\sum_t\beta^t u(d_t)$ is continuous in the product topology on $K^{\infty}$ implies (C) for the function $h \rightarrow \min_{P\in \PA}\E_P\left[ \sum\beta^t u(h_t) \right]$. The same argument shows that (C) holds also for $h \rightarrow \int \sum \beta^t u(h_t) dv$. It is left to show Comonotonic Stationarity (CS). Fix $t\in \N$, $d\in \D$ and $f,g,h\in\hh$ such that $h_t$ is comonotonic with $f_i$ and $g_i$ for all $i\geq t$. Notice that $h_t$ is comonotonic with $f_i$ if and only if for all $\omega,\omega'\in\Omega$
$$
\left[ u(h_t(\omega))-u(h_t(\omega'))\right] \left[ u(f_i(\omega))-u(f_i(\omega'))\right]\geq 0, 
$$
and the same is true for $g_i$. Let $\omega$ and $\omega'$ be in $\Omega$, we have
\begin{multline*}
\left[ u(h_t(\omega))-u(h_t(\omega'))\right] \left[ \sum_{i\geq t} \beta^{i+1} u(f_i(\omega))-\sum_{i\geq t} \beta^{i+1} u(f_i(\omega'))\right]=\\
\lim_{n\rightarrow\infty}\sum_{i=t}^n \beta^{i+1}\left[ u(h_t(\omega))-u(h_t(\omega'))\right] \left[ u(f_i(\omega))-u(f_i(\omega'))\right]\geq 0.
\end{multline*}
Therefore $h_t$ is comonotonic with $\sum_{i\geq t} \beta^{i+1} u(f_i)$ and $\sum_{i\geq t} \beta^{i+1} u(g_i)$. Denoting $\delta=\sum_{i=1}^{t-1} \beta^i u(d_i)$, and using the fact that the Choquet integral satisfies comonotonic additivity and positive homogeneity we get
\begin{align*}
(d_0,\dots,d_{t-1},h_t,f_t,f_{t+1},\dots)&\succsim (d_0,\dots,d_{t-1},h_t,g_t,g_{t+1},\dots) \Leftrightarrow\\
\int \delta+\beta^t u(h_t)+ \sum_{i\geq t} \beta^{i+1} u(f_i) dv &\geq \int \delta+\beta^t u(h_t)+ \sum_{i\geq t} \beta^{i+1} u(g_i) dv \Leftrightarrow\\
 \delta+\int \beta^t u(h_t) dv + \beta\int \sum_{i\geq t} \beta^{i} u(f_i) dv &\geq  \delta+\int\beta^t u(h_t)dv + \beta\int\sum_{i\geq t} \beta^{i} u(g_i) dv \Leftrightarrow \\
 \int \sum_{i\geq t} \beta^{i} u(f_i) dv &\geq \int \sum_{i\geq t} \beta^{i} u(g_i) dv \Leftrightarrow\\
 \int \delta+\sum_{i\geq t} \beta^{i} u(f_i) dv &\geq \int \delta+\sum_{i\geq t} \beta^{i} u(g_i) dv \Leftrightarrow\\
 (d_0,\dots,d_{t-1},f_t,f_{t+1},\dots)&\succsim (d_0,\dots,d_{t-1},g_t,g_{t+1},\dots).
\end{align*}

We turn now to sufficiency. A simple modification of Proposition \ref{prop:Koop_nest} shows that $(C)$, $(TS)$, $(M)$ and $(CS)$ imply $P$.1-5. Hence, by Proposition \ref{prop:DU} there  exists a continuous function $u:X\rightarrow \R$ and a discount factor $\beta\in(0,1)$ such that the restriction of $\succsim$ over $\D$ is represented by the functional
$$
U(d)=(1-\beta)\sum_{t=0}^{\infty}\beta^t u(d_t).
$$
Moreover $\beta\in (0,1)$ is unique and $u$ unique up to a positive affine transformation.

%

We can notice that connectedness of $X$ and continuity of $u$ imply that $u(X)$ is an interval. By Lemma 5 of Kochov \cite{Kochov} this interval has non-empty interior.  Re-normalize w.l.g. $u$ so that $[-1,1]\subseteq Range(u)$ and let $x^*\in X$ be such that $u(x^*)=0$.

\begin{lemma}\label{lemma:equiv_act}
For every $h \in \hh$ there exists $ d_h \in \D$ s.t. $ h \sim  d_h $. 
\end{lemma}
\begin{proof}
See Lemma 8 of Kochov \cite{Kochov}.
\end{proof}

Define the function $V:\hh\rightarrow \R$ by $V(h ):=U( d_h )$. Since $U$ represents preferences over $\D$, the function $V$ is well defined and represents  preferences over $\hh$.  Consider now the set 
$$
\U:=\{U\circ h:= \sum_t\beta^t u(h_t) | h:=(h_t)_t\in \hh\}.
$$
For every $h\in\hh$, $U\circ h\in\U$ is a function $U\circ h:\Omega\rightarrow \R$ and will be denoted by capital letters. Kochov \cite{Kochov} refers to these functions as \textit{util act}.

Define now the function $I:\U\rightarrow\R$ as $I(H):=V(h)$ where $h\in \hh$ is such that $U\circ h=H$. Notice that $I$ is well defined by monotonicity: if $H=U\circ h=U\circ h'$ then $h(\omega)\sim h'(\omega)$ for all $\omega\in \Omega$ , by (M) $h\sim h'$ and therefore $V(h)=V(h')$.

Recall that $\F=\cup_t\F_t$. We denote $B^o:=B^o(\Omega,\F,\R)$, i.e. the set of simple, real-valued $\F$-measurable functions over $\Omega$. Given a set $A\in \F$, $1_A\in B^o$ denotes the indicator function of the set $A$.

\begin{lemma}\label{lemma:absorbing_set}
For all $a\in B^o$, there exists $\delta>0$ such that $a\in \delta \U$, i.e. $\U$ is an absorbing subset of $B^o$.
\end{lemma}
\begin{proof}
See Lemma 9 of Kochov \cite{Kochov}.
\end{proof}


The next lemma extends $I:\U\rightarrow\R$ to $\tilde{I}:B^o\rightarrow \R$ and shows that $\tilde{I}$ is translation invariant and $\beta$-homogeneous. 

\begin{lemma}[$\tilde{I}$ is translation invariant]\label{lemma:extension_I}
There exists a unique functional $\tilde{I}:B^o\rightarrow \R$ such that the restriction $\tilde{I}|_{\U}$ of $\tilde{I}$ on $\U$ is such that $\tilde{I}|_{\U}=I$. Moreover for every $a\in B^o$, for every $\alpha\in \R$, $\tilde{I}(\beta a)=\beta\tilde{I}(a)$ and  $\tilde{I}(a+\alpha)=\tilde{I}(a)+\alpha$.
\end{lemma}
\begin{proof}
See Lemma 10, Lemma 11 and Lemma 12 of Kochov \cite{Kochov}.
\end{proof}

\begin{lemma}\label{lemma:como_add_preference}
Let $a,b,c\in B^o$ be such that $c$ is comonotonic with $a$ and $b$. Then $\tilde{I}(a)=\tilde{I}(b)\Leftrightarrow \tilde{I}(a+c)=\tilde{I}(b+c)$. 
\end{lemma}

\begin{proof}
Fix $a,b,c\in B^o$ such that $c$ is comonotonic with $a$ and $b$. Since $a,b,c$ are in $B^o$ there exists $t_1\in \N$ such that $a,b,c$ are $\F_{t_1}$-measurable. Moreover there is $t_2\in \N$ such that the range of $\beta^{t_2}a$, $\beta^{t_2}b$ and $\beta^{t_2}c$ is included in $[-1,1]$. Pick $n\geq \max\{t_1,t_2\}$ and  define for all $t\in \N$ and for all $\omega\in \Omega$
\begin{equation}\label{eq:phi(h)=a}
f_t(\omega):=\begin{cases}
x^* & \text{ if } t\neq n\\
u^{-1}(\beta^{n} a(\omega)) &  \text{ if } t=n.
\end{cases}
\end{equation}
Notice that $f \in \hh$ since it is finite (because $a$ is finite), and $u^{-1}(\beta^{n} a(\omega))$ is $\F_n$-measurable since  $\beta^{n} a(\omega)$ is $\F_t$-measurable and $\F_t\subseteq\F_n$.
In the same way define $g$ and $h$ using $b$ and $c$ respectively in the place of $a$. We have that $U\circ f= \beta^{2n}a$, $U\circ g= \beta^{2n}c$ and $U\circ h= \beta^{2n}c$. Hence
$$
\tilde{I}(a)=\tilde{I}(b)\Leftrightarrow \tilde{I}(\beta^{2n}a)=\tilde{I}(\beta^{2n}b)\Leftrightarrow I(U\circ f)=I(U\circ g) \Leftrightarrow V(f)= V(g)\Leftrightarrow f\sim g.
$$
Notice now that $h_n$ is comonotonic with $f_n$ and $g_n$ (and with $x^*$) and therefore by (CS) $f\sim g$ iff
$$
f^h:=(x^*,\dots,x^*,\underbrace{h_n}_n,\underbrace{f_n}_{n+1},x^*,\dots)\sim (x^*,\dots,x^*,\underbrace{h_n}_n,\underbrace{g_n}_{n+1},x^*,\dots)=:g^h
$$
and therefore $V(f^h)=V(g^h)$. Since $U\circ f^h=\beta^{2n}c+\beta^{2n+1}a$ and $U\circ g^h=\beta^{2n}c+\beta^{2n+1}b$ then $\tilde{I}(\beta^{2n}c+\beta^{2n+1}a)=\tilde{I}(\beta^{2n}c+\beta^{2n+1}b)$. Therefore (using Lemma \ref{lemma:extension_I}) we proved that $\tilde{I}(a)=\tilde{I}(b)\Leftrightarrow \tilde{I}(c+\beta a)=\tilde{I}(c+\beta b)$. However by Lemma \ref{lemma:extension_I} we have
$$
\tilde{I}(a)=\tilde{I}\left(\frac{\beta}{\beta}a\right)=\beta \tilde{I}\left(\frac{1}{\beta}a\right)\Leftrightarrow \tilde{I}\left(\frac{1}{\beta}a\right)=\frac{1}{\beta}\tilde{I}(a)
$$
and therefore
$$
\tilde{I}(a)=\tilde{I}(b)\Leftrightarrow \tilde{I}(\frac{1}{\beta}a)=\tilde{I}(\frac{1}{\beta}b)\Leftrightarrow \tilde{I}(c+\beta \frac{1}{\beta}a)=\tilde{I}(c+\beta \frac{1}{\beta}b)\Leftrightarrow \tilde{I}(a+c)=\tilde{I}(b+c)
$$
\end{proof}

We will prove now that $\tilde{I}$ satisfies comonotonic additivity on $B^o$.

\begin{lemma}[$\tilde{I}$ satisfies comonotonic additivity]\label{lemma:como_add_I}
Let $a,b\in B^o$ be comonotonic, then $\tilde{I}(a+b)=\tilde{I}(a)+\tilde{I}(b).$
\end{lemma}
\begin{proof}
Take $a,b\in B^o$ s.t. $a$ is comonotonic with $b$. By Lemma \ref{lemma:extension_I} (translation invariance) $\tilde{I}(a)=\tilde{I}(0+\tilde{I}(a))$. Since $b$ is comonotonic with $a$ and with the constant function $\tilde{I}(a)$,  $\tilde{I}(a+b)=\tilde{I}(\tilde{I}(a) +b)=\tilde{I}(a)+\tilde{I}(b)$ the first equality coming from Lemma \ref{lemma:como_add_preference} and the second one from  Lemma \ref{lemma:extension_I}.
\end{proof}

Let $a,b\in B^o$, then $a\geq b$ means $a(\omega)\geq b(\omega)$ for all $\omega\in \Omega$.  We will prove now that $\tilde{I}$ is monotone.

\begin{lemma}[$\tilde{I}$ is monotone]\label{lemma:monoton_I}
Let $a,b\in B^o$ be such that $a\geq b$, then $\tilde{I}(a)\geq\tilde{I}(b)$.
\end{lemma}
\begin{proof}
By Lemma \ref{lemma:absorbing_set} there is $n\in \N$ such that $\beta^n a,\beta^n b\in \U$. Let $f,g\in \hh$ be such that $U\circ f=\beta^n a$ and $U\circ g=\beta^n b$. Then $U\circ f(\omega)\geq U\circ g(\omega)$ for all $\omega\in \Omega$ and by monotonicity $f\succsim g$. Hence $V(f)\geq V(g) \Leftrightarrow \tilde{I}(\beta^n a)\geq\tilde{I}(\beta^n a)\Leftrightarrow \tilde{I}(a)\geq\tilde{I}(a)$ (where the last equivalence comes from Lemma \ref{lemma:extension_I}).
\end{proof}

We will prove now that $\tilde{I}$ is positively homogeneous.

\begin{lemma}[$\tilde{I}$ is positively homogeneous]\label{lemma:pos_homog_I}
For all $\alpha\geq 0$, for all $a\in B^o$ $\tilde{I}(\alpha a)=\alpha\tilde{I}(a)$.
\end{lemma}
\begin{proof}
This comes from Lemma \ref{lemma:como_add_I} and Lemma \ref{lemma:monoton_I} as noticed by Schmeidler \cite{Schmeidler} in Remark 1 p. 256.
\end{proof}

We proved therefore that $\tilde{I}:B^o\rightarrow \R$ satisfies comonotonic additivity (Lemma \ref{lemma:como_add_I}) and positive homogeneity (Lemma \ref{lemma:pos_homog_I}) and thus defining  defining $v(A)=\tilde{I}(1_A)$ for $A\in \F$, we can use Proposition 1 of Schmeidler \cite{Schmeidler} and we can conclude that for all $a\in B^o$
$$
\tilde{I}(a)=\int a dv.
$$
Hence for all $f,g\in \hh$,
$$
f\succsim g \text{ iff } I(U\circ f)\geq I(U\circ g) \text{ iff } \int \sum_t \beta^t u( f_t) dv \geq  \int \sum_t \beta^t u(g_t) dv.
$$
We turn to uniqueness. The fact that that $\beta$ si unique and $u$ is unique up to positive affine transformation comes from Koopmans \cite{Koop72}.\\
 Suppose now that the preference relation is represented by $J(U\circ f):=(1-\beta)\int \sum_t\beta^t u(f_t) dv'$. Fix $A\in \F$, let $x^1\in X$ be such that $u(x^1)=1$ and consider the stream $f\in \hh$ defined by
$$
f_t(\omega)=\begin{cases}
x^1 & \text{ if } \omega\in A\\
x^0 & \text{ otherwise}
\end{cases}
$$
for every $t\in\N$. Notice that $U\circ f=\frac{1_A}{1-\beta}$ and therefore $J(U\circ f)=v'(A)$ and $I(U\circ f)=v(A)$. Take $x\in X$ such that $u(x)= v(A)$ and define $g\in \hh$ by $g:=(x,x,\dots)$. Since $U\circ g=v(A)$ we obtain that $I(U\circ g)=v(A)=I(U\circ f)$ and therefore $f\sim g$. Notice that $J(U\circ g)=v(A)$, and since we supposed that $J$ represents the preference relation over $\hh$, then 
$v(A)=J(U\circ g)=J(U\circ f)=v'(A)$.
\end{proof}

\noindent \textit{\textbf{Proof of Theorem  \ref{th:EU}}}

\begin{proof}
Repeating the same steps as in Theorem \ref{th:CEU}, we can notice that the functional $\tilde{I}$ is additive on $B^o$, i.e. for all $a,b\in B^o$
$$
\tilde{I}(a+b)=\tilde{I}(a)+\tilde{I}(b).
$$
This comes from the fact that (SS) does not restrict additivity to comonotonic acts. Consider now two sets $A,B\in\F$ s.t. $A\cap B=\emptyset$. We have that $1_{A\cup B}=1_A+ 1_B$. Therefore
$$
v(A\cup B)=\tilde{I}(1_{A\cup B})=\tilde{I}(1_A+ 1_B)=\tilde{I}(1_A)+\tilde{I}(1_B)=v(A)+v(B).
$$
Which implies that $v$ is a  probability.
\end{proof}

\noindent \textit{\textbf{Proof of Theorem  \ref{th:CEU_v_convex}}}

\begin{proof}
Necessity is shown as in the proof of Theorem \ref{th:CEU}, using the fact that when $v$ is convex, $\int (a+b) dv\geq \int a dv+\int b dv$ for all $a,b\in B^o$ (see Proposition 3 of Schmeidler \cite{Schmeidler}).

We prove sufficiency. Since (PS) implies (CS) the proof of Theorem \ref{th:CEU} is valid. The only thing that is needed to show is that $v$ is convex, i.e. $v(A\cup B)+v(A\cap B)\geq v(A)+v(B)$.\\
Let $\tilde{I}:B^o\rightarrow \R$ be the functional defined in the proof of Theorem \ref{th:CEU}. 
\begin{lemma}\label{lemma:como_add_pref_convex}
Let $a,b,c\in B^o$ be such that $c$ is comonotonic with $b$. Then $\tilde{I}(a)=\tilde{I}(b)\Rightarrow \tilde{I}(a+c)\geq\tilde{I}(b+c)$. 
\end{lemma}
\begin{proof}
It follows using (PS) and doing the proof as in Lemma \ref{lemma:como_add_preference}.
\end{proof}
Let $A,B\in \F$. Notice that $\tilde{I}(1_A)=v(A)=\tilde{I}(v(A)1_{\Omega})$ and $\tilde{I}(1_B)=v(B)=\tilde{I}(v(B)1_{\Omega})$. Since $1_B$ is comonotonic with $v(A)1_{\Omega}$, by Lemma \ref{lemma:como_add_pref_convex} it follows $\tilde{I}(1_A+1_B)\geq \tilde{I}(v(A)1_{\Omega}+1_B)=v(A)+ \tilde{I}(1_B)=v(A)+v(B)$. Notice that $1_A+1_B=1_{A\cup B}+1_{A\cap B}$, and moreover $1_{A\cup B}$ and $1_{A\cap B}$ are comonotone. Hence by comonotonic additivity of the Choquet integral
$$
v(A\cup B)+v(A\cap B)=\tilde{I}(1_{A\cup B})+\tilde{I}(1_{A\cap B})=
\tilde{I}(1_{A\cup B}+1_{A\cap B})=\tilde{I}(1_A+1_B)\geq v(A)+v(B).
$$

\end{proof}

\noindent \textit{\textbf{Proof of Theorem  \ref{th:CEU_v_convex_IH}}}

\begin{proof}
The proof follows from Theorem 1 of Kochov \cite{Kochov} and the Proposition in Schmeidler \cite{Schmeidler2} p. 582. 
\end{proof}

\noindent \textit{\textbf{Proof of Proposition \ref{prop:exchange}}}

\begin{proof}
Fix $h\in \hh$. Since $h$ is bounded there exists a compact set $K_h\subseteq X$ such that $\cup_t h_t(\Omega)\subset K_h$. Since $u:X\rightarrow \R$ is continuous, we can find $M\in\R_+$ such that $\left|u(x)\right|\leq M$ for all $x\in K_h$. Hence $\forall t\in \N$ and $\forall \omega \in \Omega$, $\beta^t \left|u( h_t(\omega))\right|\leq \beta^t M$. Since $\sum_t\beta^t M$ converges to $\frac{M}{1-\beta}$, by Theorem 7.10 of Rudin~\cite{Rudin}, the series of functions $\sum_t \beta^t u(h_t)$ converges uniformly on $\Omega$. Define $H_n=\sum_{t=0}^n \beta^t u(h_t)$ and $H= \sum_{t=0}^{\infty} \beta^t u(h_t)$. Since $H_n$ converges uniformly to $H$ on $\Omega$, for every $\epsilon>0$, there exists $N\in \N$ s.t. $n\geq N$ implies $\sup_{\omega}|H_n(\omega)-H(\omega)|<\epsilon$. Hence, for $n\geq N$, using the fact that  $P$ is a probability 
$$
\left|\int H_n dP-\int H dP\right|\leq \int \left|H_n - H\right| dP<\int \epsilon dP=\epsilon
$$
where the first inequality follows from Theorem 4.4.4 (ii) and (iii) of Rao and Rao \cite{RaoRao} (notice that $H_n$ and $H$ are simple functions by the finiteness of acts). This implies that the series converges and  $\lim_n\int H_n dP=\int H dP$. Rewriting explicitly we obtain
\begin{multline*}
\sum_t\beta^t \E_P\left[u(h_t)\right]=\lim_n \sum_{t=0}^n \beta^t \int u(h_t) dP\\
=\lim_n \int\sum_{t=0}^n \beta^t  u(h_t) dP=\int \sum_t \beta^t u(h_t) dP=\E_P\left[\sum_t \beta^t u(h_t)\right].
\end{multline*}
\end{proof}

\end{document}